\documentclass[11pt,reqno]{amsart}

\usepackage{amsfonts,color}
\usepackage{amssymb,amsmath}
\usepackage[a4paper,ignoreall]{geometry}
\newtheorem{theorem}{Theorem}

\newtheorem{corollary}{Corollary}

\newtheorem{lemma}{Lemma}
\newtheorem{result}{Result}

\newtheorem{remark}{Remark}

\newtheorem{conjecture}{Conjecture}

\newcommand{\tr}{{\rm Tr}}

\newcommand{\gf}{{\mathbb F}}

\makeatletter
\newcommand{\figcaption}{\def\@captype{figure}\caption}
\newcommand{\tabcaption}{\def\@captype{table}\caption}
\makeatother

\begin{document}

\title[On a Class of Quadratic Polynomials with no Zeros and its applications]{On a Class of Quadratic Polynomials with no Zeros and its Application to APN  Functions}

\author{Carl Bracken}
\address{Department of Mathematics, School of Physical \& Mathematical Sciences, Nanyang Technology University,  Singapore.}
\email{carlbracken@ntu.edu.sg}

\author{Chik How Tan}
\address{Temasek Laboratories, National University of Singapore, 117411 Singapore}
\email{tsltch@nus.edu.sg}

\author{Yin Tan }
\address{Temasek Laboratories, National University of Singapore, 117411 Singapore}
\email{itanyinmath@gmail.com}
\maketitle

\begin{abstract}
We show that the there exists an infinite family of APN functions of the form 
$F(x)=x^{2^{s}+1} + x^{2^{k+s}+2^k} + cx^{2^{k+s}+1} +  c^{2^k}x^{2^k + 2^s} + \delta x^{2^{k}+1}$, over $\gf_{2^{2k}}$, 
where $k$ is an even integer and $\gcd(2k,s)=1, 3\nmid k$. This is actually a proposed APN family of Lilya Budaghyan and Claude Carlet who 
show in \cite{carlet-1} that the function is APN when there exists $c$ such that the polynomial 
$y^{2^s+1}+cy^{2^s}+c^{2^k}y+1=0$ has no solutions in the field $\gf_{2^{2k}}$. 
In \cite{carlet-1} they demonstrate by computer that such elements $c$ can be found over many fields, 
particularly when the degree of the field is not divisible by $3$.
We show that such $c$ exists when $k$ is even and $3\nmid k$ (and demonstrate why the $k$ odd case only re-describes
an existing family of APN functions). The form of these coefficients is given 
so that we may write the infinite family of APN functions.
\end{abstract}


\date{}
\maketitle

\section{Introduction}
Let $\gf_{2^{n}}$ be a finite field.
The number of zeros of the polynomial
\begin{equation}
  \label{Fx}
  P_a(x)=x^{2^s+1}+x+a,\ \ a\in\gf_{2^n}^*
\end{equation}
has been studied in \cite{tor-alex-1,tor-alex-2} as it has applications in several different contexts.
For example, constructing difference sets with Singer parameters (\cite{dillon,dillon-dobbertin}),
finding crosscorrelation between $m$-sequences (\cite{seq-1,seq-2}) and more recently in constructing error correcting codes \cite{BraHel}.
Also, a similar problem concerning the polynomial
$x^{p^s+1}+ax+b$ over $\gf_{p^n}$
has been considered by Bluher in \cite{bluher},
where $p$ is an odd prime.

The following result, due to Helleseth and Kholosha (\cite[Theorem 1]{tor-alex-1}) shows that, when $\gcd(s,n)=1$,
the polynomial $P_a$ can only have none, one or three zeros.

\bigskip

\begin{result}
  \label{number}
  For any $a\in\gf_{2^n}^*$ and a positive integer $l<n$ with $\gcd(l,n)=1$,
  the polynomial $P_a(x)=x^{2^l+1}+x+a$
  has either none, one, or three zeros in $\gf_{2^n}$.
  Further, let $M_i$ denotes the number of $a$ such that $P_a(x)$ has $i$ zeros, then

  (1) If $n$ is odd,
  \begin{eqnarray*}
    \begin{array}{llll}
       & M_0=\frac{2^n+1}3,   & M_1=2^{n-1}-1,   & M_3=\frac{2^{k-1}-1}{3}, \\
    \end{array}
  \end{eqnarray*}

  (2) If $n$ is even,
  \begin{eqnarray*}
    \begin{array}{llll}
       & M_0=\frac{2^n-1}3,   & M_1=2^{n-1},   & M_3=\frac{2^{n-1}-2}{3}. \\
    \end{array}
  \end{eqnarray*}
\end{result}

The condition that $P_a(x)$ has exactly one zero is also given in \cite[Theorem 1]{tor-alex-1}.
In Section II, when $n$ is even, we study the following problem.

  For which $a\in\gf_{2^n}^*$, does $P_a(x)$ have no zeros.

It is well known that $P_a(x)$ are related to other polynomials. More precisely,
for a polynomial of the form $G(x)=x^{2^s+1}+\alpha x^{2^s}+\beta x+\gamma$ over $\gf_{2^n}$, by letting $x=x+\alpha$,
$G(x)$ can be reduced to the form $H(x)=x^{2^s+1}+\alpha x+\beta$. Moreover, by a simple substitution $x=sx$ with $s^{2^i}=\alpha$,
$H(x)$ can be transformed into the form $P_a(x)=x^{2^s+1}+x+a$.
Through the above transformations and the polynomials obtained in Theorem \ref{3termpoly},
we may get the ones with the forms $G$ and $H$.
These types of polynomials are of interest in other contexts.

Before getting into the application on APN functions, we would like to briefly introduce these functions.
A function $F:\gf_{2^n}\rightarrow\gf_{2^n}$ is called \textit{almost perfect nonlinear} (APN)
if the number of solutions in $\gf_{2^n}$ of the equation
$$F(x+a)+F(x)=b$$
is at most $2$, for all $a,b\in\gf_{2^n}, a\ne 0$. We also say it has differential uniformity of 2.  APN functions were introduced by Nyberg in \cite{nyberg},
who defined them as the mappings with highest resistance to differential cryptanalysis. In other words,
APN functions are those for which the plaintext difference $x-y$ yields the ciphertext difference $f(x)-f(y)$
with probability $1/2^{n-1}$.  Since Nyberg's characterization, many new APN functions have been constructed,
see \cite{carlet-1,carl-trinomial,carl-2,carlet-2} and the references there. All the new infinite families
(from 2005) have been quadratic (algebraic degree of 2) multinomials, as is the one in this article.
When computing the differential uniformity of a quadratic function we may use the equation
$$F(x)+F(x+a)+F(a)=0.$$
This is possible as the differential equation is linear.

Another application of APN functions is in the construction of error correcting codes.
Each new APN function yields a new and inequivalent error correcting code with the parameters of the double error correcting BCH code.
In fact, APN functions are said to be inequivalent if the extended BCH-like codes constructed from them are inequivalent codes,
see \cite{carl-trinomial} for details. We refer to inequivalent APN functions as CCZ-inequivalent
(named after Carlet, Charpin and Zinoviev) \cite{ccz}. This is a more general form of equivalence than the previously used affine and extended and affine equivalences.
It is well known that CCZ equivalence preserves the differential and Walsh spectrum of the function.
Some other equivalent descriptions of CCZ-equivalence and its invariants may be found in \cite{pott}.

In \cite{carlet-1}, the authors construct a family of quadratic APN functions provided the existence
of certain polynomials. \\

\begin{result}
  \label{carletapn}
  Let $n$ and $s$ be any positive integers, $n=2k,\gcd(s,k)=1$, and $c,s\in\gf_{2^n}$
  be such that  $s\not\in\gf_{2^k}$. If the equation
  $$ G(x)=x^{2^s+1}+cx^{2^s}+c^{2^k}x+1=0$$
  has no solution $x$ such that $x^{2^k+1}=1$, and in particular if the polynomial
  $G$ is irreducible over $\gf_{2^n}$,
  then the function
  $$F_1(x)=x(x^{2^s}+cx^{2^k}+cx^{2^{k+s}})+x^{2^s}(c^{2^k}x^{2^k}+sx^{2^{k+s}})+x^{2^{k+s}+2^k}$$
  is an APN function.
\end{result}

It was verified in \cite{carlet-1} that the aforementioned irreducible polynomial $G$
exist over many fields $\gf_{2^{2k}}$ with $3\nmid k$ by a computer.
But, up to now, there is no construction of an infinitive family of such polynomials.
In Section III, when $n=2k$ with $k$ even and $3\nmid k$,
we construct the polynomials with the form $G(x)$ without zeros (Theorem \ref{newapn}) by applying the techniques used in Section II.

We will conclude this section by giving some remarks on the APN function $F_1$ in Result \ref{carletapn}.
Define the function $F:\gf_{2^n}\rightarrow\gf_{2^n}$ by
$$F(x)=x^{2^{s}+1} + x^{2^{k+s}+2^k} + cx^{2^{k+s}+1} +  c^{2^k}x^{2^k + 2^s} + \delta x^{2^{k}+1},$$
where $n=2k, \delta\not\in \gf_{2^k}$.
One may check that, if $n=2k$ with $k$ odd, $F$ is CCZ-equivalent to the multinomial APN function
in \cite[Theorem 1]{carl-trinomial} by substituting $x$ with $x+\gamma x^{2^k}$.
Also, if $k$ is even, $F$ is CCZ-equivalent to $F_1$ in Result \ref{carletapn} by substituting $x$ with $x+\gamma x^{2^k+2^s}$,
where $\gamma$ is not a $(2^k-1)$-th power.
Therefore, to prove the existence of the APN function $F_1$ in Result \ref{carletapn},
it is sufficient to consider the case $k$ is even and the existence of $F$.
One reason here to consider the function $F$ instead of the hexanomial one $F_1$ is that $F$ is equivalent to
the No. 4 function in the list of quadratic APN functions over fields with dimension $8$ in Dillon's paper \cite{dillonbanff},
and the No. 4 function is also a 5-term polynomial.
It should be mentioned that we know the family of APN functions we are constructing are inequivalent to other families,
as computer evidence has verified that it is not equivalent to any power mapping when $n=8$ \cite{carlet-1}.
Also for $n=8$ we have checked by computer that $F$ in not equivalent to $x^3 +\tr(x^9)$.
$F$ has $\Gamma$-rank $13200$ (see the definition in \cite{pott}), while $x^3 +\tr(x^9)$ has $\Gamma$-rank $13800$. Furthermore, it must be different
from the other recently discovered non power mapping APN functions as its defined on fields with different degrees.
In fact, it is a distinguishing factor of the family under consideration here that it can be defined on fields with
degrees that are powers of 2. This property is seen as desirable by some (and necessary by others) in cryptographic applications.
The only other APN functions with this property are the power mappings $x^{2^s+1}$ and $x^{2^{2s}-2^s+1}$,
known as the Gold and Kasami functions as well as $x^3 +\tr(x^9)$. These three functions are defined on fields with any degree. \\

\section{A Type of Quadratic Polynomials with no Zeros}
In this section, we will study for which $a\in\gf_{2^n}^*$, the polynomial
$P_a(x)=x^{2^s+1}+x+a$ has no zeros, where $n=2k$ and $\gcd(2k,s)=1$.

\bigskip

\begin{theorem}
  \label{3termpoly}
  Let $P_a(x)=x^{2^s+1}+x+a$ over $\gf_{2^{2k}}$, with $(2k,s)=1$. Then $P_a$ has no zeros if
  there exists a non-cube $b\in\gf_{2^{2k}}$ such that
  $$a=A(b)=\frac{b{(b+1)}^{2^{s}+2^{-s}}}{{(b+b^{2^{-s}})}^{2^s+1}}.$$
\end{theorem}

\begin{proof}
We start with the following polynomial,
$$K(x)=b{\left(x+\frac{b^{2^{-s}}+1}{b+b^{2^{-s}}}\right)}^{2^{s}+1}+{\left(x+\frac{b^{2^{-s} +1 } +b}{b+b^{2^{-s}}}\right)}^{2^{s}+1}.$$

\bigskip

Note that $b+b^{2^{-s}}$ is always not zero as $b$ is neither $1$ nor $0$.
We will demonstrate that $K$ has no zeros by setting $K(x)=0$, which gives
$$b{\left(x+\frac{b^{2^{-s}}+1}{b+b^{2^{-s}}}\right)}^{2^{s}+1}={\left(x+\frac{b^{2^{-s} +1 } +b}{b+b^{2^{-s}}}\right)}^{2^{s}+1}.$$

\bigskip

As one side of this expression is a cube, while the other is not, the only possible solutions are when its identically zero. This implies
$$x= \frac{b^{2^{-s}}+1}{b+b^{2^{-s}}}= \frac{b^{2^{-s} +1 } +b}{b+b^{2^{-s}}},$$

\bigskip

which in turn implies that $b=1$, which it's not.

Next, we expand $K(x)=0$ and gather terms to obtain,

    $$(b+1)x^{2^{s}+1}+ \left(b(\frac{b^{2^{-s}}+1}{b+b^{2^{-s}}})+ \frac{b^{2^{-s} +1 } +b}{b+b^{2^{-s}}}\right)x^{2^s} 
    +\left( b(\frac{b+1}{b+b^{2^{s}}}) + \frac{b^{2^{s} +1 } +b^{2^s}}{b+b^{2^{s}}} \right)x $$
    $$+ b\left(\frac{b^{2^{-s}}+1}{b+b^{2^{-s}}}\right)^{2^s+1}+\left(\frac{b^{2^{-s} +1 } +b}{b+b^{2^{-s}}}\right)^{2^s+1}=0. $$

\bigskip

This becomes

 $$  (b+1)x^{2^{s}+1}+(b+1)x+\frac{b(b+1)^{2^s+2^{-s}+1}}{(b+b^{2^{-s}})^{2^s+1}}=0.$$

\bigskip

Now, we divide by $b+1$ and with a few simplifications we obtain,

 $$ P_a(x)= x^{2^s+1}+x+\frac{b{(b+1)}^{2^{s}+2^{-s}}}{{(b+b^{2^{-s}})}^{2^s+1}}=0.$$

\bigskip

So $P_a(x)=0$ can not have solutions in $\gf_{2^{2k}}$ and hence $P_a(x)$ has no zeros when $a$ has the required form.
\end{proof}

\bigskip

Particularly, when $s=1$, we may show that $P_a(x)$ has no zeros if and only if $a$ is with the form in Theorem \ref{3termpoly}.
Furthermore we may claim the polynomial is irreducible as it has degree 3.

\bigskip

\begin{corollary}
  \label{sis1}
  The polynomial $P_a(x)=x^3+x+a$ is irreducible over $\gf_{2^{2k}}$, for any integer $k$,  if and only if
  $a=d+d^{-1}$ for some non-cube $d$.
\end{corollary}
\begin{proof}
  Substituting $s=1$ in $A(b)$, we have
  \begin{eqnarray*}
    A(b) &=& \frac{b(b+1)^{2+2^{-1}}}{(b+b^{2^{-1}})^{2+1}} \\
         &=& \frac{b(b^{2^{-1}}+1)^{2\cdot(2+2^{-1})}}{b^{3\cdot 2^{-1}}(b^{2^{-1}}+1)^3} \\
     &=& \frac{b+1}{b^{2^{-1}}}.
  \end{eqnarray*}
  By Theorem \ref{3termpoly}, if $a=A(b)=\frac{b+1}{b^{2^{-1}}}$, then $P_a(x)$ has no zeros. For a simpler form we let $b=d^2$ and $a$ has the form $d+d^{-1}$ for a non cube $d$.
  Conversely, by Result \ref{number}, there are $\frac{2^n-1}3$ elements $a$ such that $P_a$ has no zeros.
  Therefore, we only need to prove that $|\mbox{Im}(A)|=\frac{2^n-1}3$. Letting $C$ denote the set of non cubes in $ \gf_{2^{2k}}$,
  it can be easily verified that $A(x)=A(x^{-1})$ for any $x\in C$ and $|C|=\frac{2(2^n-1)}3$,
  so it is sufficient to prove that $A$ is 2-to-1.
  Assume there exist $x,y\in C$ such that $A(x)=A(y)$, after simplification, we get $xy(x+y)=x+y$.
  Since $x,y\ne 0$, we have $x=y$ or $xy=1$. This shows that $A$ is 2-to-1 and the proof is finished.
\end{proof}

\bigskip

Using a computer, for small values of $n$, we found
the size of the image set of $A$ is $(2^n-1)/3$.
This implies that, in Theorem \ref{3termpoly}, $P_a$ has no zeros if and only if
$a\in\mbox{Im}(A)$. However, we cannot prove this and leave this as a conjecture.

\bigskip

\begin{conjecture}
  \label{openproblem}
  Let $n=2k$ with $k$ is even and $C$ be the set of non-cubes.
  Define a function $A:C\rightarrow\gf_{2^n}$ by
  $$A(b)=\frac{b{(b+1)}^{2^{s}+2^{-s}}}{{(b+b^{2^{-s}})}^{2^s+1}}.$$
  The polynomial  $P_a(x)=x^{2^s+1}+x+a$ has no zeros if and only if $a\in\mbox{Im}(A)$.
\end{conjecture}

\bigskip

\begin{remark}
We outline a possible method for solving this problem.
It can be easily verified that $A(b)=A(b^{-1})$ and since $|C|=\frac{2(2^n-1)}{3}$, it will be sufficient to prove
$A$ is a 2-to-1 mapping.
\end{remark}

\bigskip

By Theorem \ref{3termpoly} and the relationship between the polynomials of the form
$G(x)=x^{2^s+1}+\alpha x^{2^s}+\beta x+\gamma$ and $P_a(x)=x^{2^s+1}+x+a$
mentioned in Section I, we may obtain polynomials with no zeros of the form $G$.
In the next section, we will use a variation of the method used in Theorem \ref{3termpoly}
to find the polynomial with no zeros and with the form
$x^{2^s+1}+cx^{2^s}+c^{2^k}x+1$ over $\gf_{2^{2k}}$ with $k$ is even,
for some $c \in\gf_{2^{2k}}$. It is the existence of this polynomial
that guarantees the existence of the infinite family of APN functions.
We will show that for all relevant fields a coefficient $c$ exists such that the polynomial has no zero's.

\bigskip

\section{A family of $5$-term APN functions}
Through the rest of the paper, we assume $k$ is an even integer and $3\nmid k$.
In this section, we will show that the function
\begin{equation*}
  F(x)=x^{2^{s}+1} + x^{2^{k+s}+2^k} + cx^{2^{k+s}+1} +  c^{2^k}x^{2^k + 2^s} + \delta x^{2^{k}+1}
\end{equation*}
is an APN function over $ \gf_{2^{2k}}$
by choosing a particular type of $c,\delta\in\gf_{2^{2k}}$.
The nature of the coefficients was difficult to find and we require the following lemma to show that they exist.

\bigskip

\begin{lemma}
  \label{betagammaexist}
  Let  $\mathbb{L}= \gf_{2^{2k}}$ and $k,s$ be integers such that $(s,2k)=1$ with $k$ is even and $3\nmid k$.
  Then there exist
  $\beta, \gamma \in \mathbb{L}^*$ such that
  \begin{equation}
    \label{betagamma}
    {\gamma}^{2^{s}+1}+\omega {\beta}^{2^{s}+1} + 1 =0,
  \end{equation}
  where $\omega$ has order 3 in $\mathbb{L}$  and $ {\gamma}^{2^{k}-1} \neq {\beta}^{2^{k}-1}$.
\end{lemma}

\begin{proof}
Let $\mathbb{L}^3$ denotes the set of cubes in $\mathbb{L}$. 
First, note that all non-cubes in $\mathbb{L}$ can be written as $\omega c, \omega^2 c$ for some $c \in \mathbb{L}^3$. 
We claim that there exist $a,b,c \in \mathbb{L}^3$ such that 
  \begin{equation}
    \label{abc}
    a+b=\omega c.
  \end{equation}
If this was not the case then the set of cubes would form an additive, as well as a multiplicative, 
group and hence yield a subfield of $\mathbb{L}$ with an impermissible number of elements.

Let $D$ denotes the elements of $\gf_{2^{2k}}$ which are not in $\gf_{2^{k}}$.
Now we consider the possible cubes that satisfy (\ref{abc}) with respect to whether they are in the subfield $\gf_{2^{k}}$ or 
in its complement $D$. If all solutions to  (\ref{abc}) are such that $a,b$ and $c$ are always in $D$ then the cubes in 
$\gf_{2^{k}}$ would form an additive group and yield another impossible subfield. On the other hand if all solutions are such that 
$a,b$ and $c$ are always contained in the subfield $\gf_{2^{k}}$, then the set of cubes in $D$ would form a group under 
addition with $\frac{2^{2k}-2^k}{3}$ elements. But this is not a power of $2$ and is therefore  impossible.
We conclude from the above arguments that there exists a solution $a,b,c$ to (\ref{abc}) 
with at least one element in  $\gf_{2^{k}}$ and another in $D$. 

We divide (\ref{abc}) by $a$ to obtain,
$$1 + \frac{b}{a} + \omega \frac{c}{a}=0. $$
Letting $\alpha$ be primitive in $\mathbb{L}$ and assuming that $b/a=\alpha^t, \omega c/a=\alpha^r$, we rewrite the above equation as,
  \begin{equation}
    \label{alphatr}
    1 + {\alpha}^t + {\alpha}^r=0.
  \end{equation}
  All we can say about $t$ and $r$ is that $3\mid t$ while $3\nmid r$ and that at least one of ${\alpha}^t$  and ${\alpha}^r$ is not in $\gf_{2^{k}}$.

From Equation (\ref{alphatr}) we will give two constructions of the required equation
$$ {\gamma}^{2^{s}+1}+\omega {\beta}^{2^{s}+1} + 1 =0,$$
and show that at least one of them has to allow the condition  $ {\gamma}^{2^{k}-1} \neq {\beta}^{2^{k}-1}$.

Construction 1:  We construct $\beta$  and  $\gamma$  as follows.
Let ${\alpha}^t =  {\gamma}^{2^{s}+1}$ and ${\alpha}^r=\omega {\beta}^{2^{s}+1}$. 
The existence of $\beta,\gamma$ is guaranteed by the fact that $(s,2k)=1$ and then $(2^s+1,2^{2k}-1)=3$, 
so all cubes can be represented as $(2^s+1)$-th powers. This gives the required equation (\ref{betagamma}).
Next we will assume that
${\gamma}^{2^{k}-1} = {\beta}^{2^{k}-1}.$

Therefore,  
$$ {\gamma}^{(2^s+1)(2^{k}-1)} = {\beta}^{(2^s+1)(2^{k}-1)}.$$
We now write this equation in terms of the primitive element $\alpha$ and obtain,
$${\alpha}^{t(2^k-1)} =({\omega}^2{\alpha}^r )^{2^k-1}.$$
As ${\omega}^3 =1$ and $2^k-1$ is divisible by 3, a rearrangement yields
$${\alpha}^{(t-r)(2^k-1)}=1.$$
Because ${\alpha}$ is primitive we know that $(t-r)(2^k-1) \equiv 0 \bmod  2^{2k}-1$. 
This implies that $t=r$ or ${\alpha}^{t-r} \in \gf_{2^{k}}$. The $t=r$ case is easily dismissed as $t$ and $r$ are not equivalent $\bmod 3$. The other case we save for later.

\bigskip

Construction 2:  We construct $\beta$ and $\gamma$ differently from above.
We start with (\ref{alphatr}).
Next, we divide by $ {\alpha}^t$ to obtain,
$$  {\alpha}^{-t} + 1 + {\alpha}^{r-t}=0.$$
Now, we let ${\alpha}^{-t}=  {\gamma}^{2^{s}+1}$ and ${\alpha}^{r-t}= {\omega} {\beta}^{2^{s}+1}$.
Again, we get our required Equation (\ref{betagamma}).
Now assume that
$$ {\gamma}^{2^{k}-1} = {\beta}^{2^{k}-1},$$
which implies,
$$ {\gamma}^{(2^s+1)(2^{k}-1)} = {\beta}^{(2^s+1)(2^{k}-1)}.$$
As before we write this in terms of $\alpha$ and get
  $${\alpha}^{-t(2^k-1)} =({\omega}^2{\alpha}^{r-t} )^{2^k-1}.$$
  This yields
  $${\alpha}^{r(2^k-1)}=1.$$

Note that the $\beta$ and $\gamma$ are different in this construction but the $\alpha$ is the same $\alpha$ as in Consrtuction 1.

In this case $r=0$ or ${\alpha}^{r} \in \gf_{2^{k}}$ . It is clear that $r \neq 0$. 
But if ${\alpha}^{r}$ is in $\gf_{2^{k}}$ then this construction hasn't worked. 
Just like if ${\alpha}^{t-r}$ is in $\gf_{2^{k}}$, then the first construction hasn't worked. 
But we only require one of them to work so when we assume ${\alpha}^{r} \in \gf_{2^{k}}$ 
and ${\alpha}^{t-r} \in \gf_{2^{k}}$ together we get ${\alpha}^{r} \in \gf_{2^{k}}$ and 
${\alpha}^{t} \in \gf_{2^{k}}$  which contradicts an earlier claim and the proof is complete.
\end{proof}

\bigskip

Next we will use the $\beta$  and  $\gamma$ in Lemma \ref{betagammaexist} and a modification of the techniques
in Theorem \ref{3termpoly} to prove the existence of the following family of APN functions.

\bigskip

\begin{theorem}
  \label{newapn}
  Let $s$ and $k$ be integers, with $k$ even, such that $(s,2k)=1$ and $3\nmid k$.
  We choose $\delta \notin \gf_{2^k}$, $\omega$ to have order 3 and $\beta$
  and  $\gamma$ such that  $ {\gamma}^{2^{s}+1}+\omega {\beta}^{2^{s}+1} + 1 =0 $ with  $ {\gamma}^{2^{k}-1} \neq {\beta}^{2^{k}-1}$.
  Then the function
  
  $$ F(x) = x^{2^{s}+1} + x^{2^{k+s}+2^k} + (\omega {\beta}^{2^{s}+2^k} +  {\gamma}^{2^{s}+2^k} )x^{2^{k+s}+1} 
         + (\omega {\beta}^{2^{k+s}+1} +  {\gamma}^{2^{k+s}+1} )  x^{2^k + 2^s} + \delta x^{2^{k}+1}$$

  is an APN function on $\gf_{2^{2k}}$.
\end{theorem}

\begin{proof}
The existence of the coefficients is guaranteed from Lemma \ref{betagammaexist}.
To prove the function is APN we need to show that the equation
$$F(x)+F(x+a)+F(a)=0$$
has no more than 2 solutions for all nonzero $a \in \gf_{2^{2k}}$.
Placing $F(x)$ into this equation and then replacing $x$ with $ax$ will give

\begin{eqnarray}
 \label{apn1}
 \begin{array}{lll}
 && a^{2^{s}+1}(x+x^{2^s}) + a^{2^{k+s}+2^k} (x^{2^{k+s}} +x^{2^k}) 
  + a^{2^{k+s}+1}(\omega {\beta}^{2^{s}+2^k} +  {\gamma}^{2^{s}+2^k} )(x+x^{2^{k+s}})  \\

&& \\
 && +a^{2^k + 2^s}(\omega {\beta}^{2^{k+s}+1} +  {\gamma}^{2^{k+s}+1} ) (x^{2^k}  +x^{2^s} ) 
  +a^{2^{k}+1}  \delta (x+ x^{2^k}) = 0.
 \end{array}
\end{eqnarray}

\bigskip

We raise this equation to the $2^k$-th power and obtain a second equation,

\begin{eqnarray*}
 && a^{2^{k+s}+2^k} (x^{2^{k+s}} +x^{2^k})  + a^{2^{s}+1}(x+x^{2^s}) 
 +a^{2^k + 2^s}(\omega {\beta}^{2^{k+s}+1} +  {\gamma}^{2^{k+s}+1} ) (x^{2^k}  +x^{2^s} ) \\
&&\\
 && + a^{2^{k+s}+1}(\omega {\beta}^{2^{s}+2^k} +  {\gamma}^{2^{s}+2^k} )(x+x^{2^{k+s}})  
  +a^{2^{k}+1}  {\delta}^{2^k} (x+ x^{2^k}) = 0.
\end{eqnarray*}

\bigskip

When we add these two expressions together, we simply get

$$a^{2^{k}+1}({\delta} +  {\delta}^{2^k} )(x+ x^{2^k}) = 0.$$

\bigskip

As $a \neq 0$ and $\delta \notin \gf_{2^k}$ this implies that $x \in \gf_{2^k}$.
Then our original equation (\ref{apn1}) becomes

  $$\left( a^{2^{s}+1} + a^{2^{k+s}+2^k} +  a^{2^{k+s}+1}(\omega {\beta}^{2^{s}+2^k} +  {\gamma}^{2^{s}+2^k}) \right. 
  \left. +a^{2^k + 2^s}(\omega {\beta}^{2^{k+s}+1} +  {\gamma}^{2^{k+s}+1})\right)(x+x^{2^s})=0$$

\bigskip

Clearly this will permit only 2 solutions in $x$ provided we can establish that

$$ a^{2^{s}+1} + a^{2^{k+s}+2^k} +  a^{2^{k+s}+1}(\omega {\beta}^{2^{s}+2^k} +  {\gamma}^{2^{s}+2^k})
+a^{2^k + 2^s}(\omega {\beta}^{2^{k+s}+1} +  {\gamma}^{2^{k+s}+1})  \neq 0,$$

\bigskip

for all non zero $a$.
To this end we divide by $a^{2^{s}+1}$ and let $y=a^{2^k-1}$, which gives the following,

$$G(y) = y^{2^s+1}+(\omega {\beta}^{2^{s}+2^k} +  {\gamma}^{2^{s}+2^k})y^{2^s} 
     +(\omega {\beta}^{2^{k+s}+1} +  {\gamma}^{2^{k+s}+1}) y +1$$

\bigskip

which if it has no zeros, proves $F(x)$ is APN.
To show that $G(y)$ has no zeros, we will set it equal to zero and use the techniques in Theorem \ref{3termpoly}
to produce a factoring which allows no solutions.
The $\beta$ and $\gamma$ were chosen very carefully to allow the usage of this technique.

From ${\omega}^2 {\beta}^{-2^s-1}G(y)=0$, we obtain

$$ {\omega}^2 {\beta}^{-2^s-1}y^{2^s+1} + ({\beta}^{2^k-1} +{\omega}^2 {\gamma}^{2^s +2^k} {\beta}^{-2^s-1} ) y^{2^s}  
+ ({\beta}^{2^{k+s}-2^s} +{\omega}^2 {\gamma}^{2^{k+s} +2^s} {\beta}^{-2^s-1})y    + {\omega}^2{\beta}^{-2^s-1}=0.$$

\bigskip

A simple rearrangement of ${\gamma}^{2^{s}+1}+\omega {\beta}^{2^{s}+1} + 1 =0$ will allow us to write the coefficient
of $y^{2^s+1}$ as ${\omega}^2  {(\frac{\gamma}{\beta})}^{2^{s}+1} + 1$.
The coefficient of $y^{2^s}$ can be rewritten as
\begin{equation*}
  {\omega}^2  {(\frac{\gamma}{\beta})}^{2^{s}+1} {\gamma}^{2^k-1} + {\beta}^{2^k-1},
\end{equation*}
 while the coefficient of $y$ can be written as
\begin{equation*}
  {\omega}^2  {(\frac{\gamma}{\beta})}^{2^{s}+1} {\gamma}^{2^{k+s}-2^s} + {\beta}^{{2^{k+s}}-2^s}.
\end{equation*}

Now, using the fact that,
\begin{equation*}
  {\gamma}^{2^k(2^{s}+1)}+\omega {\beta}^{2^k({2^{s}+1}) }= 1,
\end{equation*}
we can alter the last term, ${\omega}^2{\beta}^{-2^s-1}$, as follows,

\begin{eqnarray*}
 {\omega}^2{\beta}^{-2^s-1}&=& {\omega}^2{\beta}^{-2^s-1}({\gamma}^{2^k(2^{s}+1)}+\omega {\beta}^{2^k({2^{s}+1}) }) \\
  &=& {\omega}^2  {(\frac{\gamma}{\beta})}^{2^{s}+1} {\gamma}^{(2^s+1)(2^k-1)} + {\beta}^{(2^s+1)(2^k-1)}.
\end{eqnarray*}

Placing these alternate forms of the coefficients into ${\omega}^2 {\beta}^{-2^s-1}G(y)=0$ yields,

 $$ ( {\omega}^2  {(\frac{\gamma}{\beta})}^{2^{s}+1} + 1 )y^{2^s+1} + ({\omega}^2  {(\frac{\gamma}{\beta})}^{2^{s}+1} {\gamma}^{2^k-1} + {\beta}^{2^k-1})y^{2^s} 
  +({\omega}^2  {(\frac{\gamma}{\beta})}^{2^{s}+1} {\gamma}^{2^{k+s}-2^s} + {\beta}^{{2^{k+s}}-2^s})y $$

 $$ +{\omega}^2  {(\frac{\gamma}{\beta})}^{2^{s}+1} {\gamma}^{(2^s+1)(2^k-1)} + {\beta}^{(2^s+1)(2^k-1)}=0.$$

This implies
\begin{eqnarray*}
 && {\omega}^2  {(\frac{\gamma}{\beta})}^{2^{s}+1}(y^{2^s+1}+{\gamma}^{2^k-1} y^{2^s} + {\gamma}^{2^s(2^k-1)}y+{\gamma}^{(2^s+1)(2^k-1)}) \\
 && =y^{2^s+1}+{\beta}^{2^k-1} y^{2^s} + {\beta}^{2^s(2^k-1)}y+{\beta}^{(2^s+1)(2^k-1)}.
\end{eqnarray*}

Next we factor each side to obtain,
\begin{equation*}
  {\omega}^2  {(\frac{\gamma}{\beta})}^{2^{s}+1}{(y+{\gamma}^{2^k-1}) }^{2^s+1} = {(y+{\beta}^{2^k-1}) }^{2^s+1}.
\end{equation*}

Clearly the right hand side of this expression is a cube while the left hand side is not. So the only possible solutions occur when $y={\gamma}^{2^k-1}={\beta}^{2^k-1}$, but as we have chosen
$\gamma$ and $\beta$ such that ${\gamma}^{2^k-1} \neq {\beta}^{2^k-1}$, we can now say that $G(y)$ has no zeros and the proof is complete.
\end{proof}

\bigskip

We cannot determine the Fourier (Walsh) spectrum of the APN function $F$ in Theorem \ref{newapn} and we leave this as an open problem.
By a computer (with $n=8$), the Fourier spectrum of $F$ is the same as the Gold APN functions, i.e., the spectrum should takes the values $\{0, \pm 2^{n/2}, \pm 2^{(n+2)/2}  \}$.

\bigskip

All other known APN functions have had their Fourier spectra computed, see \cite{brazha} and references therein.
This is done for two reasons. The first is a cryptographic application.
Knowing a functions spectrum can allow us to compute its nonlinearity,
which measures the functions resistance to Matsui's linear attack \cite{mat}.
Secondly, the weight distribution of the codewords in the BCH-like code constructed
from the function is also determined by the functions spectrum.
All codes derived from an APN function will have a minimum distance of $5$,
but the weight distributions differ among some of the six power mapping APN functions.
The infinite families of multi-term quadratic APN functions discovered since 2005
all have the same spectrum as the gold function and we expect the function in this article to be no different.

\bigskip

\begin{conjecture}
  The Fourier spectrum of the APN function $F$ in Theorem \ref{newapn} is $\{0, \pm 2^{n/2}, \pm 2^{(n+2)/2}\}$.
\end{conjecture}

\bigskip

\section{Conclusions}
In this paper, we considered the polynomial $P_a(x)=x^{2^s+1}+x+a$ over $\gf_{2^n}$ with
$\gcd(n,s)=1$. It is shown that if $a=\frac{b{(b+1)}^{2^{s}+2^{-s}}}{{(b+b^{2^{-s}})}^{2^s+1}}$
for some non-cube $b$, then $P_a(x)=0$ has no solutions. Moreover, we conjecture that the converse of this result is also true. 
In particular, we show that $x^3+x+a$ is irreducible if and only if $a=d+d^{-1}$, for some non cube $d$.
By applying the techniques used here we obtain an infinite family of polynomials with no zero's
of the form $x^{2^s+1}+cx^{2^s}+c^{2^k}x+1$ over $\gf_{2^{2k}}$ with $k$ even and $\gcd(2k,s)=1, 3\nmid k$. 
This guarantees the existence of the infinite family of quadratic APN functions proposed by Budaghyan and Carlet in \cite{carlet-1}.

\section*{Acknowledgement}
We would like to thank Dr. Longjiang Qu for careful reading of this manuscript and  many helpful discussions.

\end{document}